\documentclass[pdftex]{article}
\pdfoutput=1
\usepackage{etex}

\usepackage{hyperref}

\usepackage{amsmath}
\usepackage{amsthm}
\usepackage{adjustbox}
\usepackage{framed}

\usepackage{graphicx}
\usepackage{framed, color} 
\usepackage[margin=0.7in]{geometry}
\usepackage{pdflscape}
\definecolor{shadecolor}{rgb}{1, 0.8, 0.3}

\usepackage{amssymb}
\usepackage{amsfonts}
\usepackage{cite}
\usepackage{epsfig}
\usepackage{subcaption}
\usepackage{array}
\usepackage{enumerate}
\usepackage{enumitem}
\usepackage[titlenumbered,ruled]{algorithm2e}
\usepackage[all]{xy}
\usepackage[compact]{titlesec}

\usepackage{times}
\usepackage{tikz}
\usepackage{float}
\usepackage{caption}
\usepackage{nth}
\usepackage{array} 

\theoremstyle{plain}

\newtheorem{lemma}{Lemma}

\newtheorem{corollary}{Corollary}

\DeclareCaptionType{Principles of Generalization}
\DeclareCaptionType{Balls-And-Bins Rules}

\theoremstyle{definition}
\newtheorem{definition}{Definition}

\theoremstyle{remark}

\newcommand{\beq}{\begin{equation}}
\newcommand{\eeq}{\end{equation}}
\newcommand{\bea}{\begin{eqnarray}}
\newcommand{\eea}{\end{eqnarray}}
\newcommand{\bean}{\begin{eqnarray*}}
\newcommand{\eean}{\end{eqnarray*}}
\newcommand{\bit}{\begin{itemize}}
\newcommand{\eit}{\end{itemize}}
\newcommand{\ben}{\begin{enumerate}}
\newcommand{\een}{\end{enumerate}}
\newcommand{\blem}{\begin{lem}}
\newcommand{\elem}{\end{lem}}
\newcommand{\bthm}{\begin{thm}}
\newcommand{\ethm}{\end{thm}}
\newcommand{\bpf}{\begin{IEEEproof}}
\newcommand{\epf}{\end{IEEEproof}}

\newcommand{\comment}[1]{}

\newcommand\defeq{\mathrel{\overset{\makebox[0pt]{\mbox{\normalfont\tiny\sffamily def}}}{=}}}

\usepackage{mathtools}
\usepackage{soul}

\usepackage{booktabs}
\usepackage{tabulary}
\usepackage{xparse}
\NewDocumentCommand{\rot}{O{45} O{1em} m}{\makebox[#2][l]{\rotatebox{#1}{#3}}}%

\usepackage{balance}

\begin{document}

\title{Density Evolution on a Class of Smeared Random Graphs: \\ A Theoretical Framework for Fast MRI}

\author{Kabir Chandrasekher, Orhan Ocal, and Kannan Ramchandran\\ 
Department of Electrical Engineering and Computer Sciences, University of California, Berkeley\\ 
\{kabirc, ocal, kannanr\}@berkeley.edu}
\date{}
\maketitle

\begin{abstract}
We introduce a new ensemble of random bipartite graphs, which we term the `smearing ensemble', where each left node is connected to some number of consecutive right nodes.
Such graphs arise naturally in the recovery of sparse wavelet coefficients when signal acquisition is in the Fourier domain, such as in magnetic resonance imaging (MRI).
Graphs from this ensemble exhibit small, structured cycles with high
probability, rendering current techniques for determining iterative decoding
thresholds inapplicable. 
In this paper, we develop a theoretical platform to analyze and evaluate the
effects of smearing-based structure.
Despite the existence of these small cycles, we derive exact density evolution
recurrences for iterative decoding on graphs with smear-length two.
Further, we give lower bounds on the performance of a much larger class from
the smearing ensemble, and provide numerical experiments showing tight agreement
between empirical thresholds and those determined by our bounds.  Finally, we describe a system architecture to recover sparse wavelet
representations in the MRI setting, giving explicit thresholds on the minimum number of Fourier samples needing to be acquired for
the $1$-stage Haar wavelet setting.
In particular, we show that $K$-sparse $1$-stage Haar wavelet coefficients of an
$n$-dimensional signal can be recovered using $2.63K$ Fourier domain samples asymptotically using $\mathcal{O}(K\log{K})$ operations.
\end{abstract}

\section{Introduction}
\label{sec:intro}
We explain our problem through an intriguing balls-and-bins game.
There are $n$ distinct colors, $d$ balls of each color and $M$ bins.
You know beforehand that only $K \ll n$ of the colors, which are selected uniformly at random from the $n$ possible colors, will be `active', but you do not know which ones they are.
You have to throw all the ($dn$) balls into the ($M$) bins.
 The rules of the game are as follows:
\begin{enumerate}[label={R\arabic*)}]
    \itemsep0em
    \item For each color $c$, you choose a subset (possibly using a randomized strategy) $B_c \subset \{0,\cdots,M-1\}$ of size $d$. Then, the system throws the balls of that color $c$ into bins $\{b + b_c: b \in B_c\}$ modulo $M$ where $b_c$ is sampled uniformly at random from $\{0,\cdots,M-1\}$\footnote{Note that the only effect of this is to randomly offset the bins for each color.}.
    \item If a bin contains a single active ball, then all $d$ balls having the same color as that ball can be removed.
    \item The process continues iteratively until either (a) all active balls have been
        removed or 
        (b) there is no bin having a single active ball.
\end{enumerate}

\noindent The goal of the game is to remove all active balls using the minimum number of
bins.
We focus on the regime in which $(n,K,M) \rightarrow \infty$, $d=\mathcal{O}(1)$ and ask the following questions:
\begin{enumerate}
    \item What is the optimal method of dispatching the $d$ balls?
        (That is, what is the optimal strategy for designing the subsets in R1?)
    \item Given $(n,K,d)$, what is the minimum number of bins ($M$) necessary?
\end{enumerate}

While this is an intriguing game in its own right, more importantly, it has connections 
to the design of
sparse-graph codes and peeling decoding.  Surprisingly, and
more relevant here, it is also intimately
related to the recovery of sparse wavelet representations from Fourier domain
samples (see Section ~\ref{sec:wavelets}).

To illustrate, suppose that $d = 3$, then the best known strategy is to throw each  ball at a bin selected uniformly at random.
It has been demonstrated that we need asymptotically $M \simeq 1.222K$ bins as $K$ grows.
This can be shown through density evolution methods, introduced by Richardson and Urbanke
in~\cite{richardson2001capacity}, which have proven powerful in analyzing the
performance of sparse-graph codes.
Now suppose that $d=6$.  The natural strategy
is to again throw each ball at a bin selected uniformly at random.
Surprisingly, this strategy is not optimal.
To see this, we give a brief introduction to the smearing ensemble.  Consider a $g$ dimensional
vector $s = [s_1, s_2, \dots, s_g]$ where $\sum_{i=1}^{g}s_i = d$.  The
ensemble is such that $g$ bins are selected at random, and for the $i$th bin,
the immediately following $s_i-1$ bins are deterministically
selected\footnote{For example, if $s_i=2$ (we henceforth refer to this as the
    smear-length), then bins $b_i$ and $b_i + 1$ (modulo $M$) are selected, where $b_i$ is
    uniformly selected on $\{0,1,\dots,M-1\}$. See Fig.~\ref{fig:balls-and-bins} for
    an illustration.}; this is what we term
the smearing ensemble.
Figure~\ref{fig:balls-and-bins} shows an example illustration for $s = [2,2]$, and we formally define the smearing ensemble in Section~\ref{sec:main_results}.  Examining Table 1, one can see that many simple
smearing strategies outperform the fully random ensemble.  In this paper, we do not claim to design
an optimal strategy for this game; rather, we provide a theoretical platform to 
analyze these structure-exploiting policies.  

\begin{figure}[h]
    \centering
    \includegraphics{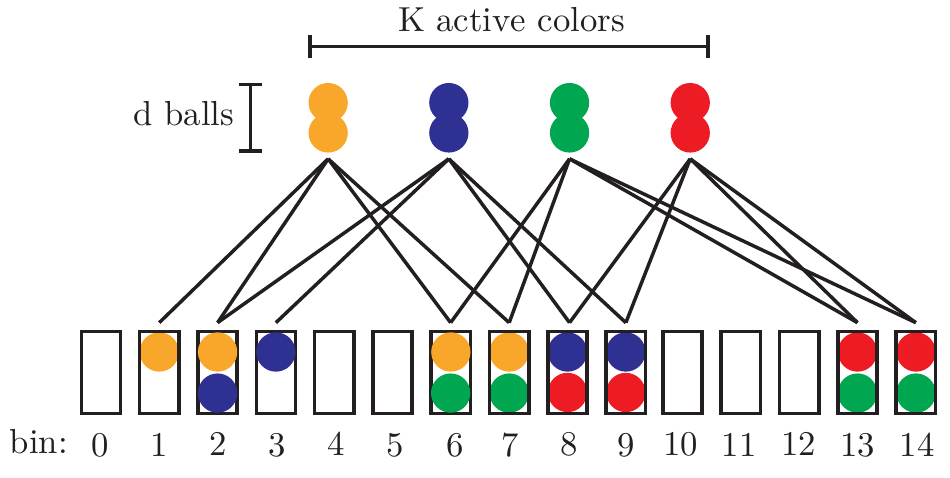}
    \caption{
An example graph showing active balls and bins that can result from rule R1.
Here $K = 4$, $M = 15$ and $d = 4$.
Each color is partitioned into $g = 2$ groups and the number of balls per group is $2$.
The game proceeds by following the rules R2-R3.
We want to find the minimum number of bins necessary to recover all the active colors.}
    \label{fig:balls-and-bins}
\end{figure}

Our balls-and-bins game is motivated by an extension of the recently
proposed FFAST (Fast Fourier Aliasing-based Sparse Transform) 
algorithm~\cite{pawar:2013} to the case where sparsity is with respect to some
wavelet basis.
In this setting, the balls correspond to the wavelet coefficients and the bins correspond to samples.
Wavelets are universally recognized to be an efficient sparse representation for the class of piecewise smooth signals having a relatively 
small number of discontinuities, a very good model for many real-world signals 
such as natural images~\cite{vetterli1995wavelets}.
In particular, we note that in MRI, images are observed to be sparse with respect to appropriately chosen wavelet bases, and acquisition is in the Fourier domain 
~\cite{lustig2008compressed}.
The computational bottleneck in recovering these images has been observed to be the computation of multiple large Fourier transforms.
Our extension of the FFAST algorithm targets this problem.
The details of this extension can be found in Section~\ref{sec:wavelets}.



\newcolumntype{L}{>{\centering\arraybackslash}m{0.45cm}}
\begin{table}[h]
    \centering
    \caption{Thresholds ($M/K$) for $d=6$.  Note that this table contains a
        strict subset of all possible strategies.}
    \label{tab:empirical_thresholds} 
    \setlength\tabcolsep{5pt}
    \begin{tabular}{l|cccccccc}
        \toprule
        Regime & \rot{$1,1,1,1,1,1$} & \rot{$1,1,1,1,2$} & \rot{$1,1,1,3$} & \rot{$1,1,4$} & \rot{$1,1,2,2$} & \rot{$1,2,3$} & \rot{$2,2,2$}\\
        \midrule
        $M/K$ & $1.570$ & $1.533$ & $1.489$ & $1.518$ & $1.533$ & $1.542$ & $1.547$ \\
        \bottomrule
    \end{tabular}
\end{table}

\subsection{Related Work}

Density evolution methods have proven powerful in analyzing the performance of
sparse-graph codes and their extensions~\cite{richardson2001capacity,richardson2008modern}.
Unfortunately, these methods apply only for sparse
random graphs that are locally tree-like.  This is not the case for all
ball-throwing strategies in the game we outlined above, e.g. the $[2,2,2]$
scheme.  
Recently Donoho et. al. have introduced approximate message passing (AMP) techniques to extend the message passing paradigm to the case when the underlying factor graph is dense~\cite{donoho2009message}.
These techniques were rigorously analyzed by Bayati and Montanari in~\cite{bayati2011dynamics}.
Although AMP has been successfully applied to many problem domains, e.g.,
\cite{tan2015compressive,maleki2013asymptotic}, it imposes a dense structure on
the factor graph. Additionally, Kudekar et.
al have been able to show the benefit of structure in convolutional LDPC codes
through the spatial coupling effect ~\cite{kudekar2011threshold, kudekar2013spatially}.
However, if the bipartite graph is sparse, but contains
small, structured cycles, it may not be necessary to invoke such methods.    



\subsection{Main Results and Organization of the Paper}
The main results of this paper are the derivation of exact thresholds for random
graphs with smear-length $2$, and bounds for higher smear-length which are empirically shown to be very tight.
We additionally detail an
application to fast recovery of sparse wavelet representations of signals when acquisition
is in the Fourier domain.  In particular, given an 
$n$-dimensional signal which is $K$-sparse\footnote{That is, an $n$-dimensional
    signal with exactly $K$ non-zero entries} with respect to the $1$-stage Haar wavelet, our analysis
shows that $2.63K$ Fourier domain samples are needed to recover the signal in
$\mathcal{O}(K\log{K})$ time.

The organization of this paper is as follows.
In Section~\ref{sec:haar}, we derive sharp thresholds for the specific case of $2$-smearing.
In Section~\ref{sec:bound}, we derive lower bounds for the probability of recovery in the case of arbitrary smearing.
We outline the connection between the ball
coloring game and the recovery of sparse wavelet representations in Section~\ref{sec:wavelets}.
We conclude with Section~\ref{sec:conclusion} by summarizing some interesting open problems and conjectures that have resulted from this work.

\section{Main Results}
\label{sec:main_results}
In this section, we introduce a new random graph ensemble, termed the `smearing ensemble', 
and show how to derive density evolution recursions for graphs from this
ensemble.  We
detail the derivation of density evolution for smearing with smear-length $2$,
and we give lower bounds for smear-length $L$.  The density evolution for
smear-length $3$ is relegated to Appendix~\ref{sec:3smear}. We now formally define the
smearing ensemble: 
\begin{definition}
\label{def:smearing}
Let $\mathcal{G}(K,M,s)$ denote the `smearing graph ensemble' with $K$ left
nodes and $M$ right nodes with connectivity characterized by $s$.
Each left node selects $g$ right nodes, where $g$ is the length of the vector $s$.
At the $i$th iteration ($0 \leq i < g$), edges are put between the left
node and right nodes $\{b_i,
b_{i+1}, \dots, b_{s_i - 1} \}$ modulo $M$ where $b_i$ is selected uniformly at random from
$\{0,1,\dots, M-1\}$.  
\end{definition}

Henceforth we are going to use the terms `left node' and `ball' interchangeably as well as `right node' and `bin' interchangeably.
We additionally refer to balls thrown to the same set of bins as a stream.  Now,
let $\lambda := gK/M$.  If $K$ balls are thrown into the $M$ bins randomly, the
degree (that is the number of balls in) of each stream will be distributed like
Poisson$(\lambda)$ by the Poisson approximation to the binomial distribution.
For the sake of brevity, we omit a general introduction to density evolution methods, pointing the interested reader to \cite{richardson2008modern}.
We now carefully derive the density evolution recurrence for smearing ensembles with the maximum smearing length of 2.

\subsection{Density Evolution on $2$-smear Graphs $(s = [2,2,2])$}
\label{sec:haar}
In our ball-coloring game, we noted that threshold for the setting $[2,2,2]$
outperforms that for the setting $[1,1,1,1,1,1]$.  In this section, we give
exact analysis for these thresholds\footnote{Although thresholds can be derived
for the other elements in the table, we give the $[2,2,2]$ derivation for
clarity.}.  First, we define our notation in Table~\ref{tab:two_thresholds}.

\begin{table}[h]
    \caption{Notation for density evolution with smear-length $2$}
    \centering %
    \begin{tabular}{l p{15cm}}
        \toprule
        $x_t:$ & Probability that a random ball is \textit{not} removed at
        iteration $t$\\
        $q_t:$ & Probability that none of the bins in a smeared pair is removed at
        iteration $t$ \\
        $d_t:$ & Probability that \textit{all} balls in the same stream as the
        reference ball are removed at time $t$ \\
        $s_t:$ & Probability that \textit{all} balls in a stream which
        intersects, but does not fully overlap, with the reference stream are
        removed at time $t$ \\
        \bottomrule
    \end{tabular}
    \label{tab:two_thresholds} 
\end{table}

\begin{figure}[!h]
    \centering
    \includegraphics[width=0.35\linewidth]{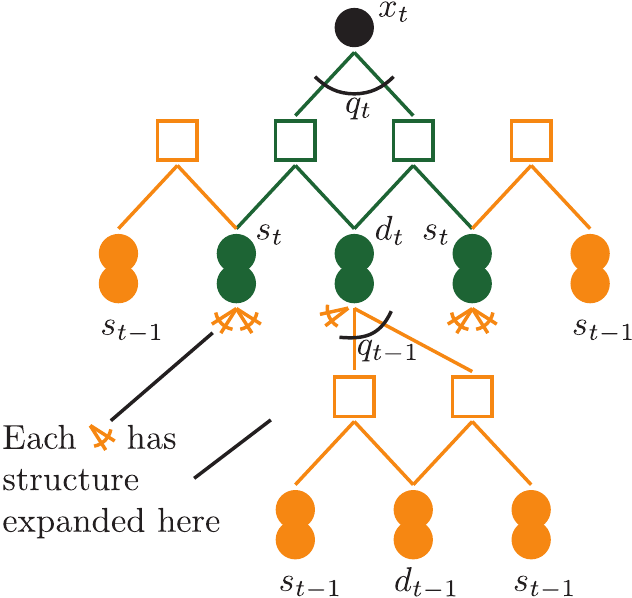}
    \caption{
Depth $2$ neighborhood of a ball in a graph with $[2,2,2]$ smearing (see
Definition~\ref{def:smearing}).
The bins and balls are grouped by their colors with respect to distance from the root.
    }
    \label{fig:haar_1_diagram}
\end{figure}

Unlike as in conventional density evolution methods in the LDPC literature which are based on an edge perspective, we take a node perspective here because dependencies between edges in the smearing setting complicate the analysis.
We refer the interested reader to~\cite{richardson2002multi} for a discussion of
the differences between node perspective and edge perspective.
Our goal is to derive recurrences for each of the quantities
defined in Table~\ref{tab:two_thresholds}.  
In this case, the equations for $x_t$ and $q_t$ are clear:
\begin{align}
x_t &= q_t^3, \\
q_t &= 1 - d_t(1 - (1 - s_t)^2).
\end{align}
To see these, we recall the dynamics of the peeling decoder: a ball is removed as soon as any of
its neighbors are removed, whereas a bin is removed only if all balls contained
in it are removed.
Thus, $x_t$ only occurs when
none of the pairs of smeared bins to which it is connected are removed.  On the
other hand, for a bin to be removed, \textit{all} of its connected nodes must be
removed.  Thus, all of the balls in the same stream as the reference ball must
be removed, which happens with probability $d_t$.  Additionally, at least one of
the two streams that do not fully overlap with the reference ball must be
removed, which happens with probability $(1 - (1 - s_t)^2)$.
Now, note that in
Fig.~\ref{fig:haar_1_diagram}, the recurring structures are those with
probabilities $s_t$ and $d_t$, so we are done upon calculating these quantities.

We first calculate $d_t$.
Let $D_1$ denote the number of balls in this stream other than the reference ball we are looking at.
It follows that $D_1$ is distributed as Binomial$(3(K-1),1/M)$, which can be approximated well for large $K$ and $M$ by Poisson($\lambda$), where $\lambda = 3K/M$.
It follows that:
\begin{align}
d_t &= \sum_{i=0}^{\infty} P(D_1 = i)(1 - q_{t-1}^2)^i \nonumber \\
    &= \sum_{i=0}^{\infty} e^{-\lambda}\frac{\lambda^{i}}{i!}(1 - q_{t-1}^2)^i
    = e^{-\lambda q_{t-1}^2}.
\end{align}

Now, we tackle $s_t$.  Note that intuitively, $s_t \ge d_t$.  This is because
each ball in a stream tracked by $s_t$ gets the same independent help from $2$
bins as $d_t$.  However, there is additionally a shared bin between all these balls.
This shared bin is able to aid in the removal of the stream tracked by $s_t$
when exactly one ball is left in the stream, and the bin has no contributions
from elsewhere.  The incorporation of this help is the key ingredient in using
the structure to help characterize the decoding thresholds.  Letting $D_2$ be the number of balls in this stream, 
we precisely characterize this as follows:
\begin{align}
s_t &= \sum_{i=0}^{\infty}P(D_2 = i)  \bigl[(1 - q_{t-1}^2)^{i} +
is_{t-1}q_{t-1}^2(1 - q_{t-2}^2)^{i-1}\bigr] \nonumber \\
    &= e^{-\lambda q_{t-1}^2} + \lambda s_{t-1} q_{t-1}^2 e^{-\lambda q_{t-2}^2}.
\end{align}
Note that the first term in this recursion is exactly $d_t$.
The second term describes the help received from the shared bin.
In order to properly characterize this term, it is necessary to introduce the
notion of \textit{memory}: the shared bin can help if it has all contributions removed except for one by time $t$.
Unlike when the neighborhood is tree-like and contains no cycles and branches
of the tree become independent~\cite{hager2015density}, when cycles are
introduced, dependence between branches is introduced.
The introduced memory
captures exactly this dependence.

We summarize the results of this section with the following lemma:
\begin{lemma}
Consider a random graph from the ensemble $G(K,M,[2,2,2])$.
Then, for $M \ge 1.547K$, recovery using the peeling decoder will succeed with high probability.
\end{lemma}

\begin{proof}
The threshold follows from the density evolution derived above.  It is important
to note that these recurrences were derived using only high probability cycles
and to be precise, it is necessary to show that the actual fraction of
unidentified balls concentrates around the average, $x_t$.  We need
to show convergence of the neighborhood of a random node in the bipartite graph
created by the independent \textit{streams}
without smearing to a tree.  This follows directly from the arguments
in~\cite{pedarsani2017phasecode} and we omit the details here. 
\end{proof}

The analogous recurrences are derived exactly for the case of smear-length $3$
in the appendix and highlight the difficulty in extending the exact analysis to
larger smear-lengths.  In the next section, we derive simple, but
effective, lower bounds for $L$-smearing.  We do this by using the following
principles of generalization, inspired by the derivation above:
\begin{table}[!h]
    \caption*{Generalizing to $L$-smearing}
    \centering %
    \begin{tabular}{l p{17cm}}
        \toprule
        1) & In a stage with $L$ smearing, there will be $L-1$ steps of memory
        necessary in order to capture the smearing structure \\
        2) & In a stage with $L$ smearing, the number of recurring structures will
        be $L$ \\
        3) & Shared bins can be used through the introduction of memory in the
        recursion \\
        \bottomrule
    \end{tabular}
\end{table}


\subsection{Lower Bound on $L$-smearing}
\label{sec:bound}
For clarity, we will consider the ensembles with $s = [L,L,L]$.
Along with $x_t, q_t, d_t$ as described in Table~\ref{tab:two_thresholds}, we define quantity $s_t^{(i)}$ in Table~\ref{tab:L_thresholds}.

\begin{table}[h]
    \caption{Notation for density evolution with smear-length $L$}
    \centering %
    \begin{tabular}{l p{15cm}}
        \toprule
        $s_t^{(i)}:$ & Probability that \textit{all} nodes in the streams which
        \textit{does not} intersect with the reference stream in $j$ bins where $1 \le j
        \le i$ 
        bins are removed at time $t$ \\
        \bottomrule
    \end{tabular}
    \label{tab:L_thresholds} 
\end{table}

\begin{figure}[!h]
    \centering
    \includegraphics[width=0.4\linewidth]{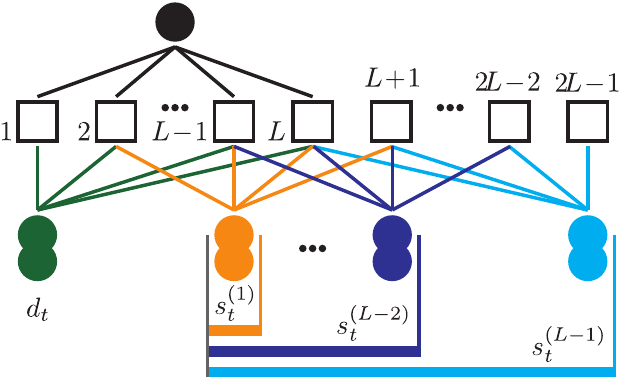}
    \caption{
Depth $1$ neighborhood of a ball under $L$-smearing.
The quantities $d_t$ and $\{s_t^{(i)}\}_{i=1}^{L-1}$ are the recurrent structures, where $s_t^{(i)}$
        tracks the joint probability that all the streams which intersect with
        the reference stream in $\{j\}_{j=L-i}^{L-1}$ bins are removed at time $t$.
    }
    \label{fig:k_smearing}
\end{figure}

As described in the principles of generalization, there will be $L$ recursions,
and up to $L-1$ memory.
The probability $d_t$ is unaffected as it depends only on the other stages.
All of the $s_t^{(i)}$, however, can
have up to $i$ memory (corresponding to the number of bins that do not
overlap with the reference ball).  Thus, we take an approach much like a
first-order approximation of a Taylor series, and allow each $s_t^{(i)}$ to use
only one step of memory.  The
following lemma characterizes the critical quantity $q_t$ in terms of its
component streams.


\begin{lemma}
In a stage with $L$-smearing,
\begin{align*}
1 - q_t &= d_t \biggl[2s_t^{(L-1)} +
\sum_{i=2}^{L-1}s_t^{(i-1)}s_t^{(L-i)} - \sum_{i=1}^{L-1}s_t^{(i)}s_t^{(L-1)} \biggr].
\end{align*}
\label{lem:pie}
\end{lemma} 

\begin{proof}
See Appendix~\ref{lemma_pie}
\end{proof}

The following lemma then establishes monotonicity of $1 - q_t$ with
respect to $d_t, s_t^{(i)}$, which we may use to complete the bound.

\begin{lemma}
Let $f(d_t, s_t^{(1)}, \dots s_t^{(L-1)}) = 1 - q_t$, then $f(d_t, s_t^{(1)}, \dots s_t^{(L-1)})$ is non-decreasing in $(d_t, s_t^{(1)}, \dots,
s_t^{(L-1)})$.  
\label{lem:monotonicity}
\end{lemma}

\begin{proof}
See Appendix~\ref{lemma:mon}.
\end{proof}

Given the above characterization of $q_t$, it now remains to give 
lower bounds for $d_t, s_t^{(1)}, \dots,
s_t^{(L-1)}$.  We can bound these as follows:
\begin{align}
d_t &= e^{-\lambda q_{t-1}^2}, \label{eq:bound_dt}
\\
s_{t}^{(i)} &\ge e^{i\lambda q_{t-1}^2} + \lambda q_{t-1}^2 e^{-\lambda
    q_{t-2}^2} r^{(i)}_{t-1}, \label{eq:bound_st}
\end{align}
for $i \in \{1,\dots,L-1\}$, where
\begin{align}
    r^{(i)}_{t} \defeq \sum_{j=1}^{i}
    \sum_{k=1}^{j}s_{t-1}^{(k)}e^{-(L-k-1)\lambda q_{t-2}^2}e^{-(k-1)\lambda
        q_{t-1}^2}.
\end{align}

There is a simple way to think about the problem so that these bounds appear.
Consider the bound on $s_{t}^{(i)}$.  This tracks the joint probability that all
the streams which intersect the reference stream in $L-j$ bins where $1 \le j
\le i$ are removed at time $t$.  In order for all of these streams to be
removed, there are two cases:

\begin{enumerate}
    \item Each stream was removed from another stage.  This probability is
        tracked by the first term: $e^{i\lambda q_{t-1}^2}$.
    \item Exactly one ball remains among all the streams.  This probability is
        tracked by the second term. 
\end{enumerate}

We focus on the second case.  Suppose that the remaining ball is in the stream
which intersects the reference stream in $j$ bins.  This implies that it is also
contained in $L-j$ shared bins that do not intersect the reference stream.  It
can be removed by any of these bins, as long as it is the only
contribution.  This help is tracked by the summation in the second
term. 

\begin{corollary}
The lower bounds given in equations~\eqref{eq:bound_dt} and \eqref{eq:bound_st} imply a lower bound
on $x_t$, the probability that a random node is removed at time $t$.
\label{cor:bound}
\end{corollary}
\begin{proof}
    This follows directly from Lemma \ref{lem:pie} and Lemma \ref{lem:monotonicity}.
\end{proof}

\begin{figure}[t]
    \centering
    \includegraphics[width=0.4\linewidth]{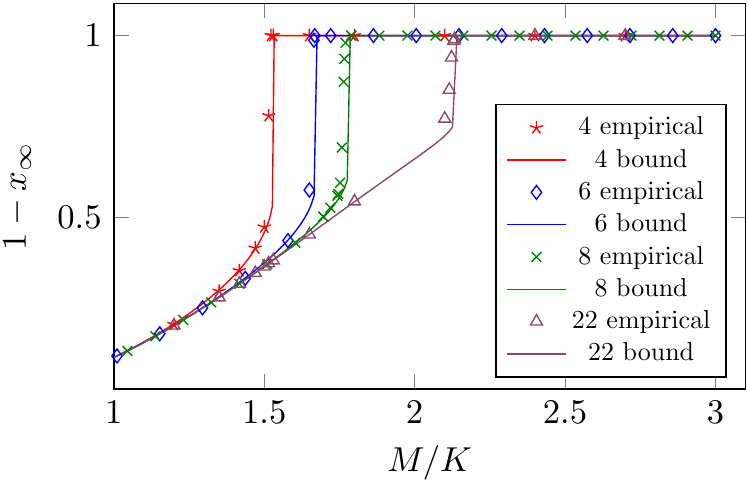}
    \caption{We plot our lower bounds against empirical simulations with
        settings $s = [1,1,L]$ for various $L$. The $y$-axis here denotes the
        probability that a random ball is removed when peeling stops (i.e., $1 - x_t$ as $t
    \rightarrow \infty$).  We note that this ensemble is
        used as it is of the most relevant practical interest to recovery of
        sparse wavelet representations (see Section~\ref{sec:wavelets}).}
    \label{fig:bounds_figs}
\end{figure}

\noindent We now give numerical experiments corroborating that the bounds in Corollary~\ref{cor:bound} capture the actual thresholds well.
We experimentally find the thresholds for full recovery by sweeping $\lambda$, and compare them to the thresholds implied by the bounds.
Fig.~\ref{fig:bounds_figs} shows these for filters with different lengths.

\section{Connections to Recovery of Sparse Wavelet Representations} 
\label{sec:wavelets}

In MRI, one acquires samples of the Fourier transform of an input signal of interest.
MRI speed is directly related to the number of samples acquired.
An inverse transform is then used to recover the original signal.
Mathematically, let $x$ be an $n$-length signal, and $X \defeq F_n x$ be its 
Fourier transform, where $F_n$ is the Fourier matrix of size $n \times n$.
In MRI, the problem is to recover $x$ from $\{X_i : i \in \mathcal{I}\}$ where the set $\mathcal{I}$ denotes the sampling locations, and this set is a design parameter.

We now present how the game of balls-and-bins and its analysis as described in Sections~\ref{sec:intro} and \ref{sec:main_results} relates to MRI.
For ease of illustration we confine ourselves to the noiseless setting and exact sparsity, but these assumptions can be relaxed.
If the signal $x$ is $K$-sparse, one can use the FFAST algorithm to recover $x$ from $O(K)$ samples with $O(K \log K)$ computations~\cite{pawar:2013,ong:2015}.
However, the images of interest in MRI are generally not sparse, but they do
have sparse wavelet representations~\cite{lustig2008compressed}. That is, we can express $x = W_n^{-1} \alpha $, where $W_n^{-1}$ is an appropriate wavelet, and $\alpha$ is sparse.
Under this signal model, the problem of recovering $\alpha$ can be transformed into the problem of decoding on an erasure channel using a sparse-graph code.
In particular, the graph for the code is drawn from a smearing ensemble with smearing length $L$ a function of the length of the underlying wavelet filter.

Furthermore, assume $\alpha$ is $K$-sparse and the length of $x$ is of the form $n = f_1 f_2 f_3$ where $f_1$, $f_2$ and $f_3$ are co-prime.
For $m \in \mathbb{Z}$ that divides $n$, let $D_{m , n}$ be the regular downsampling matrix from length $n$ to $m$, that is, $D_{m,n} = [I_m \cdots I_m]$ with $I_m$ is repeated $n/m$ times.
Let $y_{f_\ell}$ for $\ell \in \{1,2,3\}$ be the inverse Fourier transform of the downsampled $X$, that is,
    $y_{f_\ell} \defeq F_{f_\ell}^{-1} D_{f_\ell,n} X$.
Using the properties of Fourier Transform, it follows that
	$y_{f_\ell} = D_{f_\ell,n} x = D_{f_\ell,n} W_n^{-1}\alpha$.

Now, for simplicity, assume that $W$ is a block transform with block size $L$ (eg., for $1$ stage Haar wavelets $L = 2$), and the support of $\alpha$ is chosen uniformly random over the subsets of size $K$.
Using the relations between $y_{f_1}$, $y_{f_2}$ and $y_{f_3}$ and $\alpha$, recovering $\alpha$ is equivalent to decoding on a random graph from $G(K,M = f_1+f_2+f_3,s= [L,L,L])$.

We can actually `improve' the induced graph if a factor of the signal length has $L$ as a factor.
Say that $L$ divides $f_1$, it follows that $A_{f_1,n} W_n^{-1} = [W_{f_1}^{-1} \cdots W_{f_1}^{-1}]$,
where $W_{f_1}^{-1}$ is repeated $n/f_1$ times.
It can be verified that $y^\prime_{f_1} \defeq W_{f_1} y_{f_1} = A_{f_1,n} a$, hence it aliases the wavelet coefficients regularly without smearing.
The relation between $y^\prime_{f_1}$, $y_{f_2}$ and $y_{f_3}$ and $\alpha$ then induces a graph from $G(K,M = f_1+f_2+f_2,s= [1,L,L])$, which gives raise to a better threshold.

To complete the equivalence to decoding a sparse-graph code on an erasure channel, we need a mechanism to check if there is a single component in a bin (a single color in a bin).
This can be implemented by processing a shifted version of $x$ (incurring an
additional factor of $2$ of oversampling).  
We end this section with the following lemma.
\begin{lemma}
    \label{lem:haar_wavelet}
    Consider a signal $\alpha$ with ambient dimension $n$ and sparsity $K$, and access to samples from $F_{n}W_n^{-1}\alpha$. Then, the subsampling scheme described above along with the peeling
    decoder is able to exactly recover the sparse signal $\alpha$ using $2.63K$
    samples and time complexity $\mathcal{O}(K\log{K})$.
\end{lemma}
\begin{proof}
The threshold $2.63K$ follows from the density evolution derived in
Section~\ref{sec:main_results}.  The proof of complexity is given in
Appendix~\ref{sec:complexity}.
\end{proof}

\section{Conclusions and Future Work}
\label{sec:conclusion}
We have introduced a new random graph ensemble, termed the `smearing
ensemble' and devoloped a framework for deriving density evolution recurrences
for random graphs from this ensemble.  
Recalling our balls-and-bins game, our results show that some amount of smearing
can lead to a better strategy than the full random case.  A fascinating open question arises here:
what is the optimal ball-throwing strategy and what are the density evolution
recurrences for such a strategy? In this paper, we have given the first steps in
analyzing this problem rigorously.  To do this, we have leveraged the existence of small, structured
cycles and introduced the notion of memory into our density evolution.  We
believe there to be a deep connection between the introduction of memory in our
recurrences and the introduction of the `Onsager' term in the update equations
of AMP ~\cite{donoho2010message}.  We additionally believe the gains seen in
spatially coupled ensembles ~\cite{kudekar2013spatially} are intimately related 
to the structural gains of the smearing ensemble.  An extremely interesting open
problem is to determine the nature of these connections.
We have additionally shown the practical connection between the smearing ensemble and the
recovery of a sparse wavelet representation of a signal whose samples are taken
in the Fourier domain.

\balance

\bibliographystyle{IEEEtran}
\bibliography{ref}

\begin{appendix}
\section{Proofs for Bounds}
\label{sec:proofs_for_bounds}
Here, we provide proofs for the bounds on $L$-smearing.
\subsection{Proof of Lemma ~\ref{lem:pie}}
\label{lemma_pie}
\begin{proof}
    Consider the depth $1$ neighborhood of a random node, as
shown in Fig.~\ref{fig:k_smearing}.  We define the following events: let $A_i$
be the event that node $B_i$ is recovered by time $t$.  Since we are not
recursing here, we drop the references to $t$.  Thus, what we are interested in
is:
\begin{align}
q_t &= P\left(\bigcup_{i=1}^{L}P(A_i)\right) \nonumber\\
    &= P\left(A_1 \cup (A_2 \char`\\ A_1) \cup (A_3 \char`\\ A_2) \cup
    \dots \cup (A_L \char`\\ A_{L-1})
    \right)\nonumber\\
    &= P(A_1) + P(A_2 \char`\\ A_1) + \dots + P(A_L \char`\\ A_{L-1}) \nonumber \\
    &= \sum_{i=1}^{L}P(A_i) -\sum_{i=1}^{L-1}P(A_i,A_{i+1})
\end{align}
Now, we note that
$P(A_1) = P(A_k) = d_t s_t^{(L-1)}$ and $P(A_i) =
d_t s+t^{(L-i)} s_t^{(i-1)}$ where $1 < i < L$.   Additionally, we can see
that $P(A_i, A_{i+1}) = d_t s_t^{(i)} s_t^{(L-i)}$.
Plugging these into equation~\ref{eq:q_smearing} gives the result.
\end{proof}

\subsection{Proof of Lemma ~\ref{lem:monotonicity}}
\label{lemma:mon}
\begin{proof}
First, we note that 
\begin{align}
    d_t \ge s_{t}^{(j)} \ge s_{t}^{(i)}
\label{eq:mono}
\end{align}
\noindent if $i \ge j \ge 1$ by definition.  Additionally, we can see that $f(0,0,\dots, 0)
= 0$ and $f(1,1,\dots,1) = 1$.  Now, we have:
\begin{align*}
    \frac{\partial f}{\partial d_t} &= 2s_t^{(L-1)} +
    \sum_{i=2}^{L-1}s_t^{(i-1)}s_t^{(L-i)} - \sum_{i=1}^{L-1}
    s_t^{(i)}s_t^{(L-i)} \\
    &= s_t^{(L-1)} + \sum_{i=1}^{L-2} s_t^{(1)}(s_t^{(L-i-1)} - s_t^{(L-i)} ) +
    s_t^{(L-1)} - s_t^{(L-1)}s_t^{(1)} \\
    &\ge s_t^{(L-1)} + s_t^{(L-1)}(1 - s_t^{(1)}) \\
    &\ge 0 
\end{align*}
\noindent where the first inequality follows by Equation ~\ref{eq:mono} and the
last inequality follows since $s_t^{(1)} \le 1$ and $s_t^{(L-1)} \ge 0$.  Additionally,
we can see that for $1 \le i < L-1$:
\begin{align*}
    \frac{\partial f}{\partial s_t^{(i)}} = 2d_t(s_t^{(L-i-1)} -
    s_t^{(L-i)}) \ge 0
\end{align*}
\noindent where the inequality follows from Equation \ref{eq:mono}.  The
argument follows similarly for $\frac{\partial f}{\partial s_t^{(L-1)}} $. 
    Thus, since
the partial derivatives are all non-negative, the result follows.
\end{proof}

\section{$3$-smearing ($S[1,1,3]$)} 
\label{sec:3smear}
In this section, we provide exact thresholds for the ensemble drawn from
$S[1,1,3]$, illustrate why it is difficult to generalize, and draw out the
structure in the smearing patterns.

We introduce the notation $p_t$ to denote the probability of an edge from a node 
to a bin in a stage with 1 smearing
is \textit{not} removed at time $t$.
Also note that there are streams missing from the
diagram, the symmetric picture for nodes ``hanging off the edge" are not shown.
We can see that (from Lemma ~\ref{lem:pie}):
\begin{align}
q_t &= P(A_1 \cup A_2 \cup A_3) \nonumber \\
&= P(A_1) + P(A_2) + P(A_3) - P(A_1 \cap A_2)
- P(A_2 \cap A_3) - P(A_1 \cap
A_3) + P(A_1 \cap A_2 \cap A_3) \nonumber \\ 
&= d_t(2s_t^{(2)} + (s_t^{(1)})^2 - 2s_t^{(2)}s_t^{(1)})
\end{align}

Now, exactly as in section \ref{sec:haar}, we can see that:
\begin{equation}
d_t = e^{-\lambda p_{t-1}^2}
\end{equation}

Now we analyze $s_t^{(1)}$.  We omit the detail before the Taylor series
approximation for readability.  Now, note that the nodes in the stream can
either all be cleared from the other stages, or they can be cleared by the bin
shared by all streams tracked by $s_t^{(2)}$, call this bin $B_1$.  Additionally
call the bin shared by the stream thrown to the outermost bin as $B_2$.
The probability that the stream is cleared from the other stage is $ e^{-
    p_{t-1}^2}$.  The only way the stream could have been cleared by $B_1$ is
if there were exactly $1$ node remaining at time $t-1$ and $B_1$ knows it is a
singleton at time $t-1$.  Note that $B_1$ is contaminated by the stream of balls
thrown to the outermost bin and this stream must have been cleared by time $t-2$.  This can
happen in $3$ not necessarily disjoint ways.  First, $B_1$ is ``clear from
below" and all the balls in the stream thrown to the outermost bin were peeled by $t-2$.
This happens with probability 
\begin{equation}
    s_{t-1}^{(1)}e^{-\lambda p_{t-2}^2}
\label{eq:help1}
\end{equation}
Second, exactly one node is missing from the stream of balls thrown to $B_1$,
the stream of nodes thrown to the outermost bin was empty at time $t-2$, and
$B_2$ knows it
is a singleton at time $t-2$.  This happens with probability:
\begin{align}
    &\bigl(e^{-\lambda p_{t-2}^2} +
    s_{t-1}^{(2)}\lambda p_{t-2}^2e^{-\lambda
    p_{t-3}^2}\bigr) 
\cdot \lambda p_{t-2}^2e^{-\lambda p_{t-3}^2}
\label{eq:help2}
\end{align}
Finally, one ball can be missing from the stream of balls thrown to the
outermost bin.
Then, $B_2$ must be a singleton at time $t-2$, so we have the probability:
\begin{equation}
s_{t-2}^{(2)} \lambda  p_{t-2}^2 e^{-\lambda p_{t-3}^2}
\label{eq:help3}
\end{equation}
Finally, we add a term to correct for overcounting.  The overlap between the
first and last terms is:
\begin{equation}
s_{t-2}^{(2)}\lambda p_{t-2}^2e^{-2\lambda p_{t-3}^2}
\label{eq:help4}
\end{equation}
Note that if at time $t-2$, there were exactly $1$ node in the stream of nodes
thrown to $B_1$, and both $B_2$ and the bin below were singletons, then the first and
last terms overlap.  Also note that these two streams are exactly $s_{t-2}^{(2)}$ from
the perspective of $B_2$ and the bin below, and the term follows.

Thus, we can see that:
\begin{align}
s_{t}^{(1)} &= e^{-\lambda p_{t-1}^2} + \lambda  p_{t-1}^2e^{-\lambda
    p_{t-2}^2}\cdot 
\biggl[
s_{t-1}^{(1)}e^{-\lambda p_{t-2}^2} +  
\bigl(e^{-\lambda p_{t-2}^2} + s_{t-1}^{(2)}\lambda p_{t-2}^2e^{-\lambda
    p_{t-3}^2}\bigr)  
\cdot \lambda p_{t-2}^2e^{-2\lambda p_{t-3}^2} \nonumber\\
&+ s_{t-2}^{(2)} \lambda  p_{t-2}^2 e^{-\lambda p_{t-3}^2}
- s_{t-2}^{(2)}\lambda p_{t-2}^2e^{-2\lambda p_{t-3}^2}
\biggr]
\label{eq:double}
\end{align}

Finally, we give:
\begin{align}
s_{t}^{(2)} &= e^{-2\lambda p_{t-1}^2} + \lambda  p_{t-1}^2e^{-\lambda
    p_{t-2}^2}\cdot 
\biggl[
2s_{t-1}^{(1)}e^{-\lambda p_{t-2}^2} + 
\bigl(e^{-\lambda p_{t-2}^2} + s_{t-1}^{(2)}\lambda p_{t-2}^2 
\cdot e^{-\lambda
    p_{t-3}^2}\bigr)\lambda p_{t-2}^2e^{-2\lambda p_{t-3}^2} \nonumber\\
&+ s_{t-2}^{(2)} \lambda  p_{t-2}^2 e^{-\lambda p_{t-3}^2} +
s_{t-1}^{(2)}e^{-\lambda p_{t-1}^2} 
- s_{t-1}^{(2)} e^{-\lambda p_{t-1}^2} - s_{t-2}^{(2)}\lambda
p_{t-2}^2e^{-2\lambda p_{t-3}^2}
\biggr]
\label{eq:double}
\end{align}

We again note that these recursions show tight agreement with simulations. While they are difficult to digest, there is significant structure in the
recursions.  Analyzing the structure from whether or not a bin is cleared leads
us to the bounds given in Section~\ref{sec:bound}.  Additionally, one can see that whereas
$2$ smearing involved $1$ step of memory, $3$ smearing involves $2$ steps of
memory, and it becomes clear that in general $L$ smearing will involve $L-1$
steps of memory.  The amount of memory in addition to the smearing length results in complex recursions. Hence, a simple lower bound that is easy to generalize is
given in Section~\ref{sec:bound}.

\section{Proof of Lemma ~\ref{lem:haar_wavelet}}
\label{sec:complexity}
In order to describe the computational complexity, we first give pseudocode for the decoding algorithm:

\begin{algorithm}[H]
\label{alg:basis_aware}
    \SetAlgoLined
    \SetKwInOut{Input}{Input}
    \SetKwInOut{Output}{Output}
    \Input{Coupled bipartite graph $G$}
    \Output{$X$}
    \While{singletons remain}{
        \For{$B$ in bins}{
            \If{$B$ is singleton in a good stage}{
                Peel($B$)\;
            } 
            \If{$B$ is singleton in a bad stage}{
                AddToHypothesisList($B$)\;
            }
        }
    }
    \caption{Basis-aware Peeling}
\end{algorithm}

\begin{algorithm}[H]
    \label{alg:hypothesis}
    \SetAlgoLined
    \SetKwInOut{Input}{Input}
    \Input{Bin $B$}
    $L \leftarrow$ Location($B$)\;
    $V \leftarrow$ Value($B$)\;
    \For{$B^\prime$ in bins connected to variable node $(L,V)$}{
        \For{Basis $b$ in set of filters}{
            $B^\prime \leftarrow$ Hypothesis $(b,L,V)$\;
            \If{Singleton created}{
                Peel($B$);
            }
        }
    }
    \caption{AddToHypothesisList}
\end{algorithm}

Now, note that the creation of the bipartite graph is done in $\mathcal{O}(K
\log{K})$ time and is disjoint from the decoding process.  Decoding uses an
iterative decoder with a constant number of iterations
~\cite{richardson2008modern}.  Additionally, the number of bins iterated over is
linear in $K$.  We now note the size of the hypothesis list in each bin is
$\mathcal{O}(K)$.  Thus, if the list is stored using a data structure with
$\mathcal{O}(\log{K})$ insertion and deletion, such as a red-black tree, the
iterative decoding complexity is $\mathcal{O}(K\log{K})$.

\end{appendix}

\end{document}